\documentclass[a4paper]{article}
\usepackage{amsmath, amsfonts, amsthm, amssymb, bbm, bm}
\usepackage{mathrsfs}
\usepackage{centernot}
\usepackage{enumerate}
\usepackage{tabularx}
\usepackage{multicol}
\usepackage{multirow}
\usepackage[colorlinks=true,linktoc=all, allcolors=black]{hyperref}
\usepackage{setspace}
\usepackage[margin=3cm]{geometry}
\usepackage[round]{natbib}
\usepackage{booktabs}
\usepackage{tikz}
\usepackage{authblk}
\usetikzlibrary{shapes,arrows,calc}
\usepackage{bbm}

\renewcommand{\Pr}{\mathbb{P}}
\newcommand{\E}{\mathbb{E}}

\newcommand{\HT}{\text{HT}}
\newcommand{\AIPW}{\text{AIPW}}

\newcommand{\cmid}{\,|\,}

\newcommand{\T[1]}{{#1}^T}

\newcommand\indep{\protect\mathpalette{\protect\independenT}{\perp}}
\def\independenT#1#2{\mathrel{\rlap{$#1#2$}\mkern2mu{#1#2}}}

\theoremstyle{plain}
\newtheorem{lem}{Lemma}[section]

\theoremstyle{definition}
\newtheorem{dfn}[lem]{Definition}

\newcommand{\benum}{\begin{enumerate}}
\newcommand{\eenum}{\end{enumerate}}

\newcommand{\bitem}{\begin{itemize}}
\newcommand{\eitem}{\end{itemize}}

\newcommand{\barr}{\begin{array}}
\newcommand{\earr}{\end{array}}

\newcommand{\bmat}{\begin{pmatrix}}
\newcommand{\emat}{\end{pmatrix}}

\newcommand{\blist}{\renewcommand{\labelenumi}{\textbf{\arabic{enumi}}.} \begin{enumerate}}
\newcommand{\elist}{\end{enumerate} \renewcommand{\labelenumi}{\arabic{enumi}.}}

\def\bal#1\eal{\begin{align*}#1\end{align*}}

\newtheorem{theorem}{Theorem}
\newtheorem{corollary}{Corollary}

\newtheorem{lemma}{Lemma}
\newtheorem{assumption}{Assumption}

\usepackage{tikz}
\usetikzlibrary{shapes, arrows, arrows.meta, calc, positioning}

\tikzset{nv/.style={circle, color=red, fill=red, inner sep=0.5mm}}
\tikzset{rv/.style={circle, draw, thick, minimum size=7mm, inner sep=0.5mm}}
\tikzset{fv/.style={rectangle, draw, thick, minimum size=7mm, inner sep=0.5mm}}
\tikzset{lv/.style={circle, color=red, fill=gray!30, draw, thick, minimum size=7mm, inner sep=0.5mm}}
\tikzset{rve/.style={ellipse, draw, thick, minimum size=7mm, inner sep=0.5mm}}

\tikzset{rvs/.style={circle, draw, thick, minimum size=6mm, inner sep=0.5mm}}
\tikzset{fvs/.style={rectangle, draw, thick, minimum size=6mm, inner sep=0.5mm}}
\tikzset{lvs/.style={circle, color=red, fill=gray!30, draw, thick, minimum size=6mm, inner sep=0.5mm}}
\tikzset{rves/.style={ellipse, draw, thick, minimum size=6mm, inner sep=0.5mm}}

\tikzset{deg/.style={->, very thick, color=blue}}
\tikzset{degl/.style={->, very thick, color=red}}
\tikzset{beg/.style={<->, very thick, color=red}}
\tikzset{cdeg/.style={{Circle[length=+2pt 2.5,width=+2pt 2.5, fill=none]}->, very thick, color=blue}}
\tikzset{cceg/.style={{Circle[length=+2pt 2.5,width=+2pt 2.5, fill=none]}-{Circle[length=+2pt 2.5,width=+2pt 2.5, fill=none]}, very thick}}
\tikzset{uceg/.style={{Circle[length=+2pt 2.5,width=+2pt 2.5, fill=none]}-, very thick}}
\tikzset{ueg/.style={very thick}}

\definecolor{oxblue}{RGB}{0, 33, 71}

\def\var{\text{var}}
\def\vara{\text{var}_\text{a}}
\def\k{\text{K}}

\title{Strategy to select most efficient RCT samples based on observational data}
\author[1]{Wenqi Shi}
\author[2]{Xi Lin}
\affil[1]{Department of Industrial Engineering, Tsinghua University}
\affil[2]{Department of Statistics, University of Oxford}
\date{\textbf{Draft:} \today}

\begin{document}
\maketitle

\begin{abstract}
Randomized experiments can provide unbiased estimates of sample average treatment effects. However, estimates of population treatment effects can be biased when the experimental sample and the target population differ. In this case, the population average treatment effect can be identified by combining experimental and observational data. A good experiment design trumps all the analyses that come after. While most of the existing literature centers around improving analyses after RCTs, we instead focus on the design stage, fundamentally improving the efficiency of the combined causal estimator through the selection of experimental samples. We explore how the covariate distribution of RCT samples influences the estimation efficiency and derive the optimal covariate allocation that leads to the lowest variance. Our results show that the optimal allocation does  not necessarily follow the exact distribution of the target cohort, but adjusted for the conditional variability of potential outcomes. We formulate a metric to compare and choose from candidate RCT sample compositions. We also develop variations of our main results to cater for practical scenarios with various cost constraints and precision requirements. The ultimate goal of this paper is to provide practitioners with a clear and actionable strategy to select RCT samples that will lead to efficient causal inference.
\end{abstract}



\section{Introduction}
\subsection{Motivation}
There is growing interest in combining observational and experimental data to draw causal conclusions  \citep{hartman2015sate,athey2020combining,YangDing2020,Chen2021,Oberst2022,rosenman2020combining}. Experimental data from randomized controlled trials (RCTs) are considered the gold standard for causal inference and can provide unbiased estimates of average treatment effects. However, the scale of the experimental data is usually limited and the trial participants might not represent those in the target cohort. For example, the recruitment criteria for an RCT may prescribe that participants must be less than 65 years old and satisfy certain health criteria, whereas the target population considered for treatment may cover all age groups. This problem is known as the lack of transportability \citep{pearl2011transportability,rudolph2017robust}, generalizability \citep{cole2010generalizing,hernan2011compound,dahabreh2019extending}, representativeness \citep{campbell1957factors} and external validity \citep{rothwell2005external,westreich2019target}. By contrast, observational data usually has both the scale and the scope desired, but one can never prove that there is no hidden confounding. Any unmeasured confounding in the observational data may lead to a biased estimate of the causal effect. When it comes to estimating the causal effect in the target population, combining obervational and experimental data provides an avenue to exploit the benefits of both.

Existing literature has proposed several methods of integrating RCT and observational data to address the issue of the RCT population not being representative of the target cohort. \cite{kallus2018removing}, considered the case where the supports do not fully overlap, and proposed a linear correction term to approximate the difference between the causal estimates from observational data and experimental data caused by hidden confounding.\par


Sometimes even though the domain of observational data overlaps with the experimental data, sub-populations with certain traits may be over- or under-represented in the RCT compared to the target cohort. This difference can lead to a  biased estimate of the average treatment effect, and as a result, the causal conclusion may not be generalizable to the target population. In this case, reweighting the RCT population to make it applicable
to the target cohort is a common choice of remedy \citep{hartman2015sate,andrews2017weighting}. In particular, Inverse Probability of Sampling Weighting (IPSW) has been a popular estimator for reweighting \citep{cole2008constructing,cole2010generalizing,stuart2011use}. In this paper, we base our theoretical results 
on the IPSW estimator. 






\subsection{Design Trumps Analysis}
Most of the existing literature, including those discussed above, focuses on the analysis stage after RCTs are completed, and propose methods to analyse the data as given. This means, the analysis methods, including reweighting through IPSW, are to passively deal with the RCT data as they are. However, the quality of the causal inference is largely predetermined by the data collected. `Design trumps analysis' \citep{rubin2008objective}; a carefully designed experiment benefits the causal inference by far more than the analysis that follows. Instead of marginally improving through analyses, we focus on developing a strategy for the design phase, specifically the selection of RCT participants with different characteristics 
, to fundamentally improve the causal inference.\par

When designing an RCT sample to draw causal conclusions on the target cohort, a heuristic strategy that practitioners tend to opt for is to construct the RCT sample that looks exactly like a miniature version of target cohort. For example, suppose that we want to examine the efficacy of a drug on a target population consisting $30\%$ women and $70\%$ men. If the budget allows us to recruit 100 RCT participants in total, then the intuition is to recruit exactly $30$ females and $70$ males. This intuition definitely works, yet, is it efficient? We refer to the efficiency of the reweighted causal estimator for the average treatment effect in the target population, and specifically, its variance \footnote[1]{We note that the efficiency of an unbiased estimator  $T$ is formally defined as $e(T) = \frac{\mathcal{I}^{-1}(\theta)}{var(T)}$, that is, the ratio of its lowest possible variance over its actual variance. For our purpose, we do not discuss the behaviour of the Fisher information of the data but rather focus on reducing the variance of the estimators. With slight abuse of terminology, in this paper, when we say that one estimator is more efficient than another, we mean that the variance of the former is lower. Similarly, we say that an RCT sample is more efficient if it eventually leads to an estimator of lower variance.}.

In fact, we find that RCTs following the exact covariate distribution of the target cohort do not necessarily lead to the most efficient estimates after reweighting. Instead, our result suggests that the optimal covariate allocation of experiment samples is the target cohort distribution adjusted by the conditional variability of potential outcomes. Intuitively, this means that an optimal strategy is to sample more from the segments where the causal effect is more volatile or uncertain, even if they do not make up a large proportion of the target cohort.

\subsection{Contributions}
In this work, we focus on the common practice of generalizing the causal conclusions from an RCT to a target cohort. We aim at fundamentally improving the estimation efficiency by improving the selection of individuals into the trial, that is, the allocation of a certain number of places in the RCT to individuals of certain characteristics. We derive the optimal covariate allocation that minimizes the variance of the causal estimate of the target cohort. Practitioners can use this optimal allocation as a guide when they decide `who' to recruit for the trial. We also formulate a deviation metric that quantifies how far a given RCT allocation is from optimal, and practitioners can use this metric to decide when they are presented with several candidate RCT allocations to choose from.\par
We develop variations of the main results to cater for various practical scenarios such as where the total number of participants in the trial is fixed, or the total recruitment cost is fixed while unit costs can differ, or with different precision requirements: best overall precision, equal segment precision or somewhere in between. In this paper, we provide practitioners with a clear strategy and versatile tools to select the most efficient RCT samples.

\subsection{Outline} 
The remainder of this paper is organized at follows: In Section~\ref{sec:1setup}, we introduce the problem setting, notations, provide the main assumptions and provide more details on the IPSW estimator that we consider. In Section~\ref{sec:2res}, we derive the optimal covariate allocation for RCT samples to improve estimation efficiency, propose a deviation metric to assess candidate experimental designs and illustrate how this metric influences estimation efficiency. Section~\ref{sec:3estimate} provides an estimate of the optimal covariate allocation and the corresponding assumptions to ensure consistency. Section~\ref{sec:4pras} extends the main results and propose design strategies under other practical scenarios like heterogeneous unit cost and same precision requirement. In Section~\ref{sec:5numerical}, we use two numerical studies, a synthetic simulation and a semi-synthetic simulation with real-word data, to corroborate our theoretical results.



\section{Setup, Assumptions and Estimators} \label{sec:1setup}

\subsection{Problem Setup and Assumptions}
In this paper, we based our notations and assumptions on the potential outcome framework \citep{rubin1974estimating}. We assume to have two datasets: a RCT and an observational data. We also make the assumption that the target cohort of interest is contained in the observational data. \par
Define $S \in \{0,1\} $ as the sample indicator where $s = 1$ indicates membership of the experimental data and $s = 0$ the target cohort, where $T \in \{0,1\}$ as the treatment indicator and $t = 1$ indicates treatment and $t = 0$ indicates control. Let $Y_{is}^{(t)}$ denotes potential outcome for a unit $i$ assigned to data set $s$ and treatment $t$. We define $X$ as a set of observable pre-treatment variables, which can consist discrete and/or continuous variables. Let $n_0$, $n_1$, $n = n_0 + n_1$ denote the number of units in the target cohort, RCT, and the combined dataset, respectively. We use $f_1(x)$ and $f_0(x)$ to denote the distribution of $X$ in the RCT population and target cohort, respectively. 

The causal quantity of interest here is the average treatment effect (ATE) on the target population, denoted by $\tau$ .
\begin{dfn} (ATE on target cohort)
$$
 \tau := \E \left [Y^{(1))} - Y^{(0)} \mid S = 0 \right].
$$
\end{dfn}
We also define the CATE on the trial population, denoted by $\tau(x)$.
\begin{dfn} (CATE on trial population)
$$
 \tau(x) := \E \left [Y^{(1)} - Y^{(0)} \mid X = x, S = 1 \right].
$$
\end{dfn}

To ensure an unbiased estimator of the ATE on the target population after reweighting the estimates from the RCT, we need to make several standard assumptions.
\begin{assumption}(Identifiability of CATE in the RCT data)
\label{assump::identifiability1}
For all the observations in the RCT data, we assume the following conditions hold.
\begin{itemize}
    \item[(i)] 
    Consistency: $Y_{i} = Y_{i1}^{(t)}$ when $T=t$ and $S=1$;
     \item[(ii)] Ignorability:  $Y_{i}^{(t)} \indep T \mid (X, S =1)$;
     \item[(iii)] Positivity: $0 < \mathbb{P}(T=t \mid X, S =1) < 1$ for all $t \in \{0,1\}$.
\end{itemize}
\end{assumption}
The ignorability condition assumes that the experimental data is unconfounded and the positivity condition is guaranteed to hold in
conditionally randomized experiments.The igonrability and positivity assumptions combined is also referred to as strong ignorability. Under Assumption~\ref{assump::identifiability1}, the causal effect conditioned on $X = x$ in the experimental sample can be estimated without bias using:
\begin{eqnarray*}
\hat \tau(x) &=& \frac{ \sum_{S_i=1,X_i = x} \frac{T_i Y_i}{e(x)} - \frac{(1-T_i) Y_i }{ 1-e(x)}}{\sum_{S_i=1,X_i = x} 1},
\end{eqnarray*}
where $e(x) = \mathbb{P}(T=1 \mid X=x, S=1)$ is the probability of treatment assignment in the experimental sample. This estimator is also known as the Horvitz-Thompson estimator \citep{horvitz1952generalization}, which we will provide more details later in this section.

To make sure that we can `transport' the effect from the experimental data to the target cohort, we make the following transportability assumption.
\begin{assumption}(Transportability) 
\label{assump::transport}
$Y^{(t)} \indep S \mid (X, T=t)$.
\end{assumption}
Assumption~\ref{assump::transport} can be interpreted from several perspectives, as elaborated in \cite{hernan2010causal}. First, it assumes that all the effect modifiers are captured by the set of observable covariates $X$. Second, it also ensures that 
the treatment $T$ for different data stays the same. If the assigned treatment differs between the study population and the target population, then the magnitude of the causal effect of treatment will differ too. Lastly, the transportability assumption prescribes that there is no interference across the two populations. That is, treating one individual in one population does not interfere with the outcome of individuals in the other population.\par


Furthermore, we require the trial population fully overlaps with the the target cohort, so that we can reweight the CATE in the experimental sample to estimate the ATE in the target cohort. That is, for each individual in the target cohort, we want to make sure that we can find a comparable counterpart in the experimental sample with the same characteristics. 

\begin{assumption}
(Positivity of trial participation) 
\label{assump::positivity}
$0 < \mathbb{P}(S=1 \mid T=t, X = x) < 1$ for all $x \in \text{supp}(X\cmid S=0) $.
\end{assumption}
In Assumption~\ref{assump::positivity}, $\text{supp}(X\cmid S=0)$ denotes the support of the target cohort, in other words, the set of values that $X$ can take for individuals in the target cohort. Mathematically, $x \in \text{supp}(X\cmid S=0)$ is equivalent to $\mathbb{P}\left(\|X-x\| \leq \delta \mid S =  0\right)>0$, $\forall \delta>0$. Assumption~\ref{assump::positivity} requires that the support of the experimental sample includes the target cohort of interest.

\subsection{Estimators and related work}
Inverse Propensity (IP) weighted estimators were proposed by \cite{horvitz1952generalization} for surveys in which subjects are sampled with unequal probabilities.
\begin{dfn}(Horvitz-Thompson estimator)
\begin{align*}
    \widehat Y^{(t)}_{\HT} &=  \frac{1}{n_1} \sum_{i=1}^{n_1} \frac{I(T_i=t) Y_i}{\mathbb{P}\left(T_i = t \cmid X = X_i \right)},\\ 
    \hat{\tau}_{\text{HT}} &=  
    \widehat Y^{(1)}_{\HT} -  
    \widehat Y^{(0)}_{\HT} = \frac{1}{n_1} \sum_{i=1}^{n_1} \frac{T_i Y_i}{e\left(X_i\right)} - \frac{(1- T_i) Y_i}{1- e\left(X_i\right)},
\end{align*}
\end{dfn}
where the probability of treatment $e(X_i)$ is assumed to be known as we focus on the design phase of experiments. In practice, we can extend the Horvitz-Thompson estimator by replacing $e(x)$ with an estimate $\hat{e}(x)$, for example the Hajek estimator \citep{hajek1971comment} and the difference-in-means estimator.

\begin{dfn} (Augmented Inverse Propensity Weighted estimator)
\begin{eqnarray*}
   \widehat Y^{(1)}_{\AIPW} &=&\frac{1}{n_1} \sum_{i=1}^{n_1}\left[\hat{m}^{(1)}\left(X_i\right)+\frac{T_i}{\hat{e}\left(X_i\right)}\left(Y_i-\hat{m}^{(1)}\left(X_i\right)\right)\right], \\
  \widehat Y^{(0)}_{\AIPW}&=&\frac{1}{n_1} \sum_{i=1}^{n_1}\left[\hat{m}^{(0)}\left(X_i\right)+\frac{(1-T_i)}{1- \hat{e}\left(X_i\right)}\left(Y_i-\hat{m}^{(0)}\left(X_i\right)\right)\right], \\
  \hat{\tau}_{\text{\text{AIPW}}} & =& \frac{1}{n_1} \sum_{i=1}^{n_1} \frac{T_i(Y_i-\hat{m}^{(1)}(X_i))}{\hat{e}\left(X_i\right)} - \frac{(1-T_i)(Y_i-\hat{m}^{(0)}(X_i)))}{1-\hat{e}\left(X_i\right)} + \hat{m}^{(1)}(X_i) - \hat{m}^{(0)}(X_i),
\end{eqnarray*}
where $m^{(t)}(x)$ denotes the average outcome of treatment $t$ given covariate $X = x$, that is, $m^{(t)}(x)=\E [Y \mid T=t, X=x, S =1]$, and $\hat{m}^{(t)}(x)$ is an estimate of $m^{(t)}(x)$ \citep{robins1994correcting}. 
\end{dfn}
The estimator $\hat{\tau}_{\text{AIPW}}$ is doubly robust: $\hat{\tau}_{\text{AIPW}}$ is consistent if either (1) $\hat{e}\left(X_i\right)$ is consistent or (2) $\hat{m}^{(t)}(x)$ is consistent.

\begin{dfn}(Inverse Propensity Sample Weighted (IPSW) estimator)
\begin{eqnarray*}
\hat{\tau}_{\text{IPSW}}^* &=& \frac{1}{n_1} \sum_{i \in \{i:S_i = 1\}} w\left(X_i\right){\left(\frac{Y_i A_i}{e(X_i)}-\frac{Y_i\left(1-A_i\right)}{1-e(X_i)}\right)}, 
\\
w(x) &=& 
\frac{f_0(X)}{f_1(X)}.
\end{eqnarray*}
\end{dfn}
We can see that the IPSW estimator extends the Horvitz-Thompson estimator by adding a weight $w(x)$, which is the ratio between the probably of observing an individual with characteristics $X = x$ in the trial population that in the target population \citep{stuart2011use}\footnote[2]{The definition of the weight $w(x)$ differs slightly from that in \cite{stuart2011use}, where $w(x)$ is defined as $\frac{P\left(S =1 \mid X=x\right)}{P\left(S = 0 \mid X = x\right)}$. That is, the ratio of the distribution of being selected into the trail over being selected into the target cohort. This definition is based on the problem setting where there is a super population which the target cohort and trial cohort are sampled from. Our definition here agrees with that in \cite{colnet2022reweighting}.}. We use an asterisk in the notation to denote that it is an oracle definition where we assume both $f_1(X)$ and $f_0(X)$ are known, which is probably unrealistic.  The IPSW estimator of average treatment effect on target cohort is proven to be unbiased under Assumptions~\ref{assump::identifiability1}--~\ref{assump::positivity}.\par 
A concurrent study of high relevance to our work by \cite{colnet2022reweighting} investigated performance of IPSW estimators. In particular, they defined different versions of IPSW estimators, where $f_1(X)$ and $f_0(X)$ are either treated as known or estimated, and derived the expressions of asymptotic variance for each version. They concluded that the semi-oracle estimator, where $f_1(X)$ is estimated and $f_0(X)$ is treated as known, outperforms the other two versions giving the lowest asymptotic variance.

\begin{dfn}\label{dfn:semioracleIPSW}(Semi-oracle IPSW estimator, \cite{colnet2022reweighting})
\begin{eqnarray*}
\hat{\tau}_{\text{IPSW}} &=& \frac{1}{n_1} \sum_{i \in \{i:S_i = 1\}} \frac{f_0(X_i)}{\hat{f_1}(X_i)}{\left(\frac{Y_i A_i}{e(X_i)}-\frac{Y_i\left(1-A_i\right)}{1-e(X_i)}\right)}, \text{ and} \\
\hat{f_1}(x) &=& \frac{1}{n_1} \sum_{S_i=1} \mathbbm{1}_{X_i = x}.
\end{eqnarray*}
\end{dfn}

 The re-weighted ATE estimator we use in this paper, $\hat{\tau}$, coincides with their semi-oracle IPSW estimator defined above, where $f_1(X)$ is estimated from the RCT data.

\section{Main Results}
\label{sec:2res}
In this section, we start with the case where the number of possible covariate values is finite and derive the optimal covariate allocation of RCT samples that minimizes the variance of the ATE estimate, $\hat{\tau}$. We then develop a deviation metric, $\mathcal{D}(f_1)$, that quantifies how much a candidate RCT sample composition with covariate distribution $f_1$ deviates from the optimal allocation. We prove that this deviation metric, $\mathcal{D}(f_1)$, is proportional to the variance of $\hat{\tau}$ therefore it can be used as a metric for selection. Finally, we derive the above results in presence of continuous covariates.
\label{sec:theoreticalresults}

\subsection{Variance-Minimizing RCT Covariate Allocation}
We first consider the more straight-forward case, where the number of possible covariate values is finite. 
Recall that $e(x)$ denotes the propensity score,
We assume that the exact value of $e(x)$ is known for the RCT.

When units in the experimental dataset cover all the possible covariate values, for $m=1,\ldots,M$, recall the Horvitz-Thompson inverse-propensity weighted estimators \citep{horvitz1952generalization} of CATE:
\begin{eqnarray*}
\hat \tau(x_m) &=& \frac{ \sum_{S_i=1,X_i = x_m} \frac{T_i Y_i}{e(x_m)} - \frac{(1-T_i) Y_i }{ 1-e(x_m)}}{\sum_{S_i=1,X_i = x_m} 1}.
\end{eqnarray*}

Discrete covariates can be furthered divided into two types: ordinal, for example, test grade,  and categorical such as  blood type. For ordinal covariates, we can construct a smoother estimator by applying kernel-based local averaging:
\begin{eqnarray*}
    \hat \tau_\k (x_m) = \frac{\frac{1}{n_1 h^k} \sum_{S_i=1} \left( \frac{T_i Y_i}{ e (X_i)} - \frac{(1-T_i) Y_i}{1- e (X_i)}\right) K \left(\frac{X_i -x_m}{h}\right)}{\frac{1}{n_1  h^k} \sum_{S_i = 1}  K\left(\frac{X_i -x_m}{h}\right)},
\end{eqnarray*}
where $K(\cdot)$ is kernel function and $h$ is the smoothing parameter. Conceptually, the kernel function measures how individuals with covariates in proximity to $x_m$ influence the estimation of $\hat \tau_\k (x_m)$.
This kernel-based estimator works even if the observational data does not fully overlap with the experimental data. The estimator $\hat\tau_\k$ is inspired by \cite{abrevaya2015cate}, who used it to estimate the CATE. Specifically, if the covariate is ordinal and the sample size of a sub-population with a certain covariate value is small or even zero, we can consider $\hat \tau_\k(x)$, as it applies local averaging so that each CATE is informed by more data.

To study the variance of CATE estimates $\hat \tau(x_m)$ and $\hat \tau_\k (x_m)$, we define the following terms:
\begin{eqnarray*}
    \sigma_\psi^2(x) &=& \E\left[ \left( \psi(X,Y,T) - \tau(x) \right)^2 \mid X=x, S=1 \right],\\
    \psi(x,y,t) &=& \frac{t(y-m^{(1)}(x))}{e(x)} - \frac{(1-t)(y-m^{(0)}(x))}{1-e(x)} + m^{(1)}(x) - m^{(0)}(x).
\end{eqnarray*}
The random vector $\psi(X,Y,T)$ is the influence function of the AIPW estimator \citep{bang2005doubly}. Term $\sigma_\psi^2(x)$ measures the conditional variability of the difference in potential outcomes given covariate $X = x$, and $m^{(t)}(x)$ denotes the average outcome with treatment $t$ given covariate $X = x$.

\begin{assumption}\label{con::clt}
As $n$ goes to infinity, $n_1/n$ has a limit in $(0,1)$.
\end{assumption}
Assumption~\ref{con::clt} suggests that when we consider the asymptotic behavior of our estimators, sample sizes for both experimental data and observational data go to infinity, though usually there is more observational samples than experimental samples.

\begin{theorem}\label{thm::fclt}
Under Assumption~\ref{con::clt}, for $m=1,\ldots, M$, we have
\begin{eqnarray*}
\sqrt {n_1} (\hat \tau(x_m) - \tau(x_m)) &\stackrel{d}{\rightarrow}& N \left( 0, \frac{\sigma_\psi^2(x_m)}{f_1(x_m)} \right),\\
\sqrt{n_1  h} (\hat \tau_\k (x_m) - \tau(x_m)) &\stackrel{d}{\rightarrow}&
\mathcal{N} \left( 0, \frac{\Vert K \Vert_2^2 \sigma_\psi^2(x_m) }{f_1(x_m)} \right),
\end{eqnarray*}
where $\Vert K \Vert_2 = (\int K(u)^2 du)^{1/2}$.
\end{theorem}

Theorem~\ref{thm::fclt} shows the asymptotic distribution of the two CATE estimators for every possible covariate value. Complete randomization in experiments ensures that $\hat \tau(x)$ is unbiased. Based on the idea of IPSW estimator, we then construct the following two reweighted estimators for ATE:
$$
    \hat \tau = \sum_{m=1}^M f_0(x_m) \hat \tau(x_m), \quad
    \hat \tau_\k = \sum_{m=1}^M f_0(x_m) \hat \tau_\k(x_m).
$$
It is easy to see that the $\hat{\tau}$ above is the same as the semi-oracle IPSW estimator defined in Definition~\ref{dfn:semioracleIPSW} once we substitute in the expression of $\hat{\tau}(x_m)$.


\begin{theorem}\label{thm::covreal}
Under Assumption~\ref{assump::identifiability1}--\ref{con::clt}, we have
\begin{eqnarray*}
    n_1 \var(\hat \tau) &=& \sum_{m=1}^M f_0^2(x_m)\frac{\sigma_\psi^2(x_m)}{f_1(x_m)},\\
    n_1 h \vara(\hat \tau_\k) &=& \Vert K \Vert_2^2 \sum_{m=1}^M f_0^2(x_m)\frac{\sigma_\psi^2(x_m)}{f_1(x_m)},
\end{eqnarray*}
where $\vara(\cdot)$ denotes the asymptotic variance. For $m=1, \ldots, M$, the optimal covariate RCT distribution to minimize both $\var(\hat \tau)$ and $\vara(\hat \tau_\k)$ is 
\begin{eqnarray*}
    f_1^*(x_m) = \frac{f_0(x_m) \sigma_\psi(x_m)}{\sum_{j=1}^M f_0(x_j) \sigma_\psi(x_j) }.
\end{eqnarray*}
\end{theorem}
Theorem~\ref{thm::covreal} indicates that even if the covariate distribution of experimental data is exactly the same as that of the target cohort, it does not necessarily produce the most efficient estimator. The optimal RCT covariate distribution also depends on the conditional variability of potential outcomes. In fact, $f_1^*$ is essentially the target covariate distribution adjusted by the variability of conditional causal effects. This result suggests that we should sample relatively more individuals from sub-populations where the causal effect is more volatile, even if they do not take up a big proportion of the target cohort. Moreover, the two estimators, $\hat{\tau}$ and $\hat \tau_\k$, share the same optimal covariate weight no matter whether local averaging is applied. 

In practice, if the total number of samples is fixed, experiment designers can select RCT samples with covariate allocation identical to $f_1^*$ to improve the efficiency of IPSW estimate.
 
\subsection{Deviation Metric}

\begin{corollary}\label{cor::Dmetrix}
Under Assumption~\ref{assump::identifiability1}--\ref{con::clt}, we have
\begin{eqnarray*}
    n_1 \var(\hat \tau) &= & 
    \left( \sum_{m=1}^M f_0(x_m) \sigma_\psi(x_m)  \right)^2 \times
    \left\{ \mathcal{D}(f_1) + 1\right\}
    \\
    n_1 h \vara(\hat \tau_\k) &= & 
    \Vert K \Vert_2^2 \left( \sum_{m=1}^M f_0(x_m) \sigma_\psi(x_m)  \right)^2 \times
    \left\{ \mathcal{D}(f_1) + 1\right\},
\end{eqnarray*}
where $\var_{1}(\cdot) = \var_{X \mid S=1}(\cdot)$, and we define
$$
\mathcal{D}(f_1) = \var_{1} \left( \frac{f_1^*(X)}{f_1(X)}\right)
$$
as the deviation metric of experiment samples as it measures the difference between the optimal covariate distribution $f_1^*$ and the real covariate distribution $f_1$. We have $\mathcal{D}(f_1) \geq 0$, and $\mathcal{D}(f_1^*) = 0$ if and only if the real covariate distribution of experiment samples is identical to the optimal one, i.e. ${f_1^*(x)} = {f_1(x)}$ for $\forall x \in \{ x_1, \ldots, x_M \}$.
\end{corollary}

Accoring to Corollary~\ref{cor::Dmetrix}, the variance of $\hat \tau$ depends on two parts: the first part $ \sum_{m=1}^M f_0(x_m) \sigma_\psi(x_m)$ depends on the true distribution of the target population, while the second part $\mathcal{D}(f_1)$ is a measure of the deviation of the RCT sample allocation $f_1$ compared to the optimal variability-adjusted allocation $f_1^*$, and can thus reflect the representativeness of our RCT samples.
As Corollary~\ref{cor::Dmetrix} shows, the variance of IPSW estimator for the population, $\hat \tau$, is proportional to $\mathcal{D}(f_1)$. 

The deviation metric equips us with a method to compare candidate experiment designs. To be specific, if experiment designers have several potential plans for RCT samples, they can choose one with the smallest deviation metric to maximize the estimation efficiency.
 
\subsection{Including Continuous Covariates}

For continuous covariates, for instance, body mass index (BMI), we apply stratification based on propensity score. By considering an appropriate partition of the support $\{A_1, \ldots, A_L\}$ with finite $L \in \mathbb{N}$, we can turn it into the discrete case above. 

\begin{assumption}\label{con::cont}
For $l = 1, \ldots, L$, and $x, x^\prime \in A_l$, we have
\begin{itemize}
    \item[(i)] $\Pr(T=1 \mid X=x) = \Pr (T=1 \mid X = x^\prime)$;
    \item[(ii)] $\E(Y^{(1)} -Y^{(0)} \mid X=x, S=1 ) = \E(Y^{(1)} -Y^{(0)} \mid X=x^\prime, S=1)$.
\end{itemize}
\end{assumption}

Assumption~\ref{con::cont} assumes that units within each stratum share the same propensity score and CATE. This is a strong but reasonable condition if we make each stratum $A_l$ sufficiently small. Under Assumption~\ref{con::cont}, let $\hat\tau(A_l)$, $\hat\tau_\k(A_l)$, $\sigma_\psi^2(A_l)$, $e(A_l)$ denote the causal effect estimate, causal effect estimate with kernel-based local averaging, variance of influence function, propensity score, that are conditioned on $X \in A_l$. Let $f_0(A_l) = \Pr(X \in A_l \mid S=0)$ and $f_1(A_l) = \Pr(X \in A_l \mid S=1)$. We can then construct two IPSW estimators:
$$
    \hat \tau = \sum_{l=1}^L f_0(A_l) \hat \tau(A_l), \quad
    \hat \tau_\k = \sum_{l=1}^L f_0(A_l) \hat \tau_\k(A_l).
$$

As shown in Corollary~\ref{thm::covc}, we have similar results to Theorem~\ref{thm::covreal}, but instead of the optimal covariate distribution, we derive the optimal probability on each covariate set $A_l$.

\begin{corollary}\label{thm::covc}
Under Assumption~\ref{assump::identifiability1}--~\ref{con::cont}, we have
\begin{eqnarray*}
    n_1 \var(\hat \tau) &=& \sum_{l=1}^L f_0^2(A_l)\frac{\sigma_\psi^2(A_l)}{f_1(A_l)},\\
    n_1 h \vara(\hat \tau_\k) &=& \Vert K \Vert_2^2 \sum_{l=1}^L f_0^2(A_l)\frac{\sigma_\psi^2(A_l)}{f_1(A_l)}.
\end{eqnarray*}
For $l=1, \ldots, L$, the optimal distribution on each covariate set to minimize both $\var(\hat \tau)$ and $\vara(\hat \tau_\k)$ is 
\begin{eqnarray*}
    f_1^*(A_l) = \frac{f_0(A_l) \sigma_\psi(A_l)}{\sum_{j=1}^L f_0(A_j) \sigma_\psi(A_j) }.
\end{eqnarray*}
Moreover,
\begin{eqnarray*}
    \sum_{l=1}^L f_0^2(A_l)\frac{\sigma_\psi^2(A_l)}{f_1(A_l)}
      &=& 
    \left( \sum_{l=1}^L f_0(A_l) \sigma_\psi(A_l)  \right)^2 \times
    \left\{ \mathcal{D}(f_1) + 1\right\},
\end{eqnarray*}
where $A(x) = \{ A: x \in A; A \in \{A_1, \ldots, A_L\}\}$.
\end{corollary}

In the sections that follow, for simplicity, we illustrate our method in the scenario where the covariates are all discrete with finite possible values. The results can easily be extended to include continuous covariates following the same logic as descibed in this section.

\section{Estimating Conditional Variability}
\label{sec:3estimate}
The optimal covariate allocation derived above can benefit the planning of the composition of RCT samples. However, it's difficult, or impossible, to estimate the conditional variability of potential outcomes prior to RCTs being carried out. In this section, we provide a practical strategy using information from the observational data to estimate the theoretical optimal covariate distribution, and derive conditions under which our strategy yields consistent results.

In completely randomized experiments, we can show that
\begin{eqnarray*}
\sigma_\psi^2(x) = \frac{1}{e(x)} \var(Y^{(1)} \mid X=x) + \frac{1}{1-e(x)} \var(Y^{(0)} \mid X=x).
\end{eqnarray*}
To estimate $\sigma_\psi^2(x)$ by observational data, $\forall x$, let

\begin{align*}
\widehat Y^{(0)}(x) &= \frac{\sum_{S_i=0,T_i=0} Y_i}{\sum_{S_i=0,T_i=0} 1 }, \quad&
\widehat Y^{(1)}(x) &= \frac{\sum_{S_i=0,T_i=1} Y_i}{\sum_{S_i=0,T_i=1} 1 },\\
\widehat S^{(0)}(x) &= \frac{\sum_{S_i=0,T_i=0} \left(Y_i - \widehat Y^{(0)}(x) \right)^2}{\sum_{S_i=0,T_i=0} 1 - 1}, \quad&
\widehat S^{(1)}(x) &= \frac{\sum_{S_i=0,T_i=1} \left(Y_i - \widehat Y^{(1)}(x) \right)^2}{\sum_{S_i=0,T_i=1} 1 - 1}.
\end{align*}

We can then estimate the conditional variability of potential outcomes, the optimal covariate distribution and the deviation metric of RCT samples from observational data as follows
\begin{eqnarray*}
\hat{\sigma}_\psi^2(x) &=& \frac{1}{e(x)}\widehat S^{(1)}(x) + \frac{1}{1-e(x)}\widehat S^{(0)}(x)\\
\hat f_1^*(x_m) &= &\frac{f_0(x_m) \hat\sigma_\psi(x_m)}{\sum_{j=1}^M f_0(x_j) \hat\sigma_\psi(x_j) },\\
\widehat{\mathcal{D}}(f_1) &=& \var_{1} \left( \frac{\hat f_1^*(X)}{f_1(X)}\right).
\end{eqnarray*}

Assumption~\ref{cond::prop} below ensures the consistency of the estimated conditional variance of potential outcomes $\hat{\sigma}_\psi^2(x)$. The main problem of estimating $\hat{\sigma}_\psi^2(x)$ from observational data is the possibility of unobserved confounding. Instead of assuming unconfoundedness, 
our Assumption~\ref{cond::prop} is weaker and requires that the expectation of estimate $\hat{\sigma}_\psi^2(x)$ 
is proportional to the target conditional variability, which is a weaker condition.
\begin{assumption}\label{cond::prop}
For $\forall x$, suppose
\begin{eqnarray*}
&& \frac{1}{e(x)}\var{(Y^{(1)} \mid X=x, T=1, S=0)} + \frac{1}{1-e(x)}\var{(Y^{(0)} \mid X=x, T=0, S=0)} \\
&=& c \left[ \frac{1}{e(x)}\var{(Y^{(1)} \mid X=x, S=0)} + \frac{1}{1-e(x)}\var{(Y^{(0)} \mid X=x, S=0)} \right],
\end{eqnarray*}
where $c>0$ is an unknown constant.
\end{assumption}
The left-hand side of the equation above is the conditional variance of observed outcomes that can be estimated from the observational data, and the right-hand side is the theoretical conditional variance of potential outcomes that we want to approximate. Assumption~\ref{cond::prop} requires these two quantities to be proportional, rather than absolutely equal. Intuitively, the assumption supposes that the covariate segments in the observational data that exhibit high volatility in observed outcomes, also have high variance in their potential outcomes, although it is not required that the absolute levels of variance have to be the same.

\begin{theorem}\label{thm::asyocd}
Under Assumption~\ref{con::clt} and Assumption~\ref{cond::prop}, $\forall x$ we have
$$
    \hat f_1^*(x) \rightarrow f_1^*(x), \quad \widehat{\mathcal{D}}(f_1) \rightarrow \mathcal{D}(f_1).
$$
Thus, $\hat f_1^*(x)$ and $\widehat{\mathcal{D}}(f_1)$ are consistent. 
\end{theorem}

Based on Thereom~\ref{thm::asyocd}, we propose a novel strategy to select efficient RCT samples. Specifically, we select the candidate experimental design with covariate allocation $f_1$ that minimizes the estimate of deviation metric $\widehat{\mathcal{D}}(f_1)$. By contrast, a naive strategy prefers the candidate experimental design with $f_1 = f_0$, which mimics exactly the covariate distribution in the target cohort. If the conditional variability of potential outcomes $\sigma_\psi^2(x)$ vary widely according to $x$. Our strategy can lead to much more efficient treatment effect estimator compared to the naive strategy.

\section{Practical Scenarios}
\label{sec:4pras}
\subsection{Heterogeneous Unit Cost}
We also consider the experimental design with a cost constraint and heterogeneous cost for different sub-populations. The goal is to find the optimal sample allocation for RCT that minimizes the variance of the proposed estimator subject to a cost constraint. For $m = 1, \ldots, M$, let $C_m$ denote the cost to collect a sample in the sub-population with $X = x_m$. 

\begin{theorem}\label{thm::cost}
Under the cost constraint that
$$
    \sum_{m=1}^M C_m \left( \sum_{S_i=1, X_i = x_m} 1 \right) = C,
$$
the optimal sample allocation $f_1(X)$ that minimizes $\var(\hat \tau)$ is
$$
    f_1^{\text{c},*} (x_m) = \frac{f_0(x_m)\sigma_\psi(x_m)/\sqrt{C_m}} {\sum_{i=1}^M f_0(x_i)\sigma_\psi(x_i)/\sqrt{C_i}}
$$
for $m = 1, \ldots, M$. Here we use the superscript $c$ to denote the cost constraint.
\end{theorem}

Theorem~\ref{thm::cost} suggests that the optimal RCT covariate allocation under the given cost constraint is the covariate allocation of target cohort adjusted by both the heterogeneous costs for sub-populations and the conditional variability of potential outcomes. Intuitively, compared to the case without heterogeneous costs, we should include more RCT samples with lower unit cost.

\subsection{Different Precision Requirements}
Theorem~\ref{thm::covc} shows the optimal sample allocation to maximize the efficiency of average treatment effect estimator for the target cohort. If we require the same precision for estimators in each domain, we need the sample allocation as follows:
$$
    f_1^{\text{s},*}(x_m) = \frac{ \sigma^2_\psi(x_m)}{\sum_{j=1}^M  \sigma^2_\psi(x_j) },
$$
where we use the superscript s to denote the requirement of same precision for the CATE estimate in each segment.

Intuitively, to take both objectives into consideration, we propose a compromised allocation that falls between the two optimum allocations $\forall k \in [0,1]$:
$$
    f_1^{k,*}(x_m) = \frac{f_0^k(x_m) \sigma^{2-k}_\psi(x_m)}{\sum_{j=1}^M f_0^k(x_j) \sigma^{2-k}_\psi(x_j) }.
$$

\begin{corollary}\label{prop::1}
If for $m = 1, \ldots, M$, 
$$
    f_0(x_m) = \frac{\sigma_\psi(x_m)}{\sum_{j=1}^M \sigma_\psi(x_j)},
$$
we have for $\forall k \in [0,1]$,
$$
    f_1^*(X) = f_1^{\text{s},*}(X) = f_1^{k,*}(X).
$$

The deviation metric for sample allocation under same precision strategy and compromise strategy are
\begin{eqnarray*}
    \mathcal{D}(f_1^{\text{s},*}) &=& \var_1\left( \frac{f_0(X)}{\sigma_\psi(X)} \right) \left( \frac{\sum_{m=1}^M \sigma^2_\psi(x_m)}{\sum_{m=1}^M f_0(x_m) \sigma_\psi(x_m)} \right)^2,\\
    \mathcal{D}(f_1^{k,*}) &=& \var_1\left( \frac{f_0(X)}{\sigma_\psi(X)} \right)^{1-k} \left( \frac{\sum_{m=1}^M f_0^{k}(x_m) \sigma^{2-k}_\psi(x_m)}{\sum_{m=1}^M f_0(x_m) \sigma_\psi(x_m)} \right)^2,
\end{eqnarray*}
respectively.
\end{corollary}

\section{Numerical Study}
\label{sec:5numerical}
\subsection{Simulation}
In this section, we conduct a simulation study to illustrate how representativeness of experiment samples influences the estimation efficiency of average treatment effect for a target cohort, and demonstrate how the representativeness metric $\mathcal{D}(f_1)$ can facilitate the selection from candidate RCT sample designs. We set the size of observational dataset $n_0 = 10000$ and the size of experimental dataset $n_1 = 200$. For the units in the observational data, we draw covariates $x$ from $\{1, 2, 3\}$ with probability 0.3, 0.2 and 0.5 respectively. We then set
\begin{align*}
    Y^{(0)} = 2X + X^4\epsilon, \quad Y^{(1)} = 1 - X  + \epsilon, \text{ then,}\\
    Y^{(1)} - Y^{(0)} = 1 -3X + (X^4 -1)\epsilon,
\end{align*}
where $\epsilon \sim \mathcal{N}(0,1)$. We can then compute the conditional variability of potential outcomes $\sigma^2_\psi(x)$ and thus the optimal covariate distribution $f^*_1$ from the true population. Our model engenders distinctive conditional variability $\sigma^2_\psi(x)$ given different $x$, making the optimal covariate distribution $f_1^*$ very different from the target covariate distribution $f_0$.

For experimental data, we simulate 100 different candidate experimental sample designs. In each design, we randomly draw experiment samples from the target cohort with probability
$$
    \Pr(S=1 \mid X=x) = \frac{e^{p_x}}{e^{p_1} + e^{p_2} + e^{p_3}},
$$
where $p_1, p_2, p_3$ are i.i.d. samples drawn from standard normal distribution. We can then compute the real covariate distribution $f_1(x)$ and the repressiveness metric $\mathcal{D}(f_1)$. To estimate the efficiency of average treatment effect estimator, we conduct 1000 experiments for each design. In each experiment, the treatment for each unit follows a Bernoulli distribution with probability 0.5. The simulation result is shown in Fig~\ref{fig:sim}. The relationship between the variance and representativeness can be fit into a line, which is consistent with our result that $\var(\hat \tau) \propto \mathcal{D}(f_1)$. The red line shows the value $R(f_1)$ for the naive strategy mimicking exactly the target cohort distribution, which is not zero, and we can see that it is not the optimal RCT sample and does not produce the most efficient causal estimator. 


\begin{figure}
    \centering
    \includegraphics[width = 0.5\textwidth]{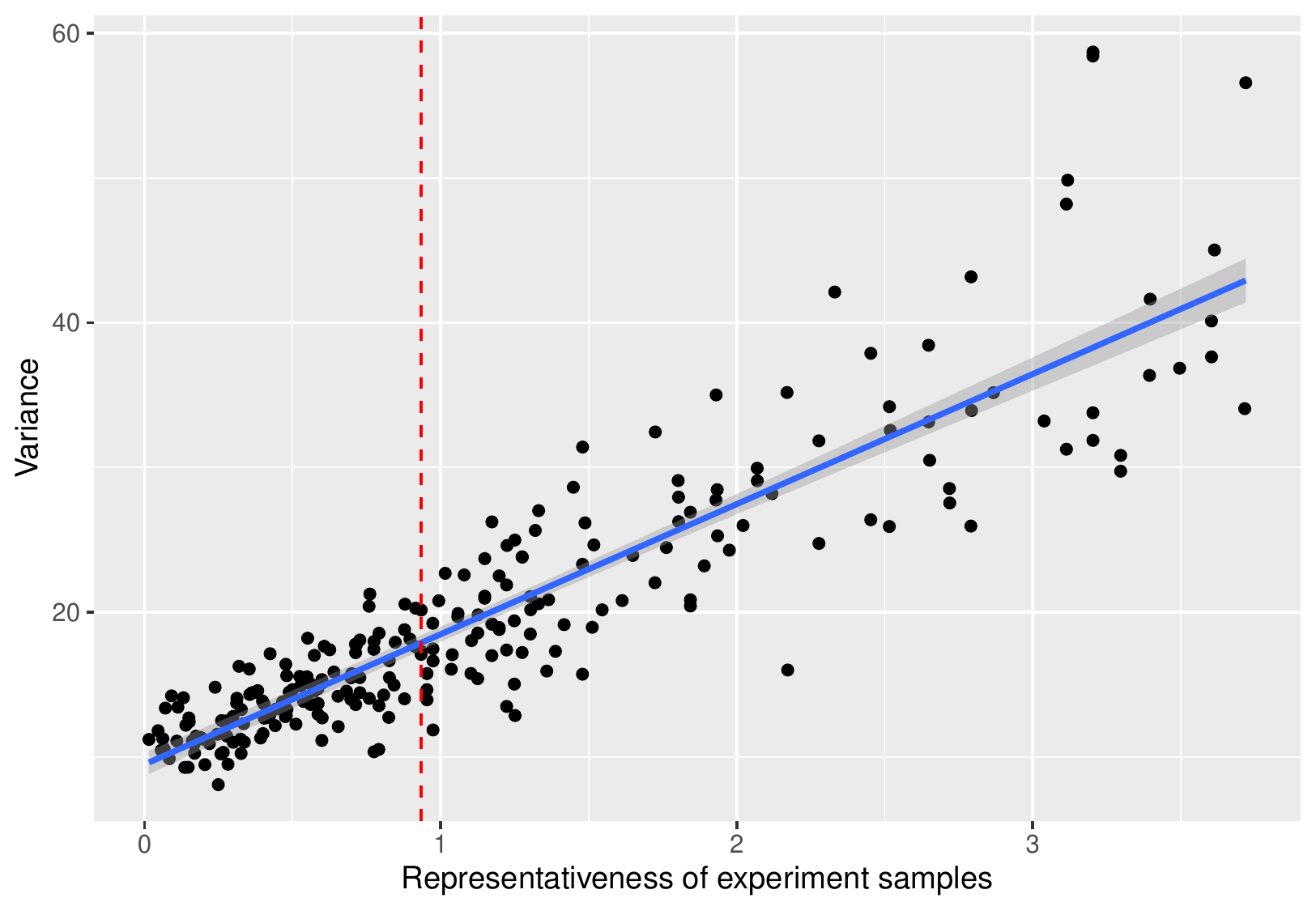}
    \caption{How the deviation metric of experiment samples $\mathcal{D}$ correlates with the estimated variance of the $\hat{\tau}$. The red line marks the deviation metric $\mathcal{D}$ of a trial sample selected following the na\"ive strategy.}
    \label{fig:sim}
\end{figure}

For the case with heterogeneous unit cost, we set the costs for sub-populations with covariate $X$ being 1, 2, 3 to be 20, 30, 40, respectively. The total capital available is 30000. Instead of randomly drawing experiment samples from the target cohort with a fixed number of total subjects, we randomly assign capitals for different sub-populations with a fixed amount of total cost. Given the budget assigned to each sub-population, we then draw subjects randomly from the sub-population, where the amount of subjects can be determined by the assigned budget. The simulation result is illustrated by Figure~\ref{fig:simc}. We can see that under the cost constraint, experiment samples that follow a distribution closer to $f_1^{c,*}$ lead to causal estimator with better efficiency.

\begin{figure}
    \centering
    \includegraphics[width = 0.5\textwidth]{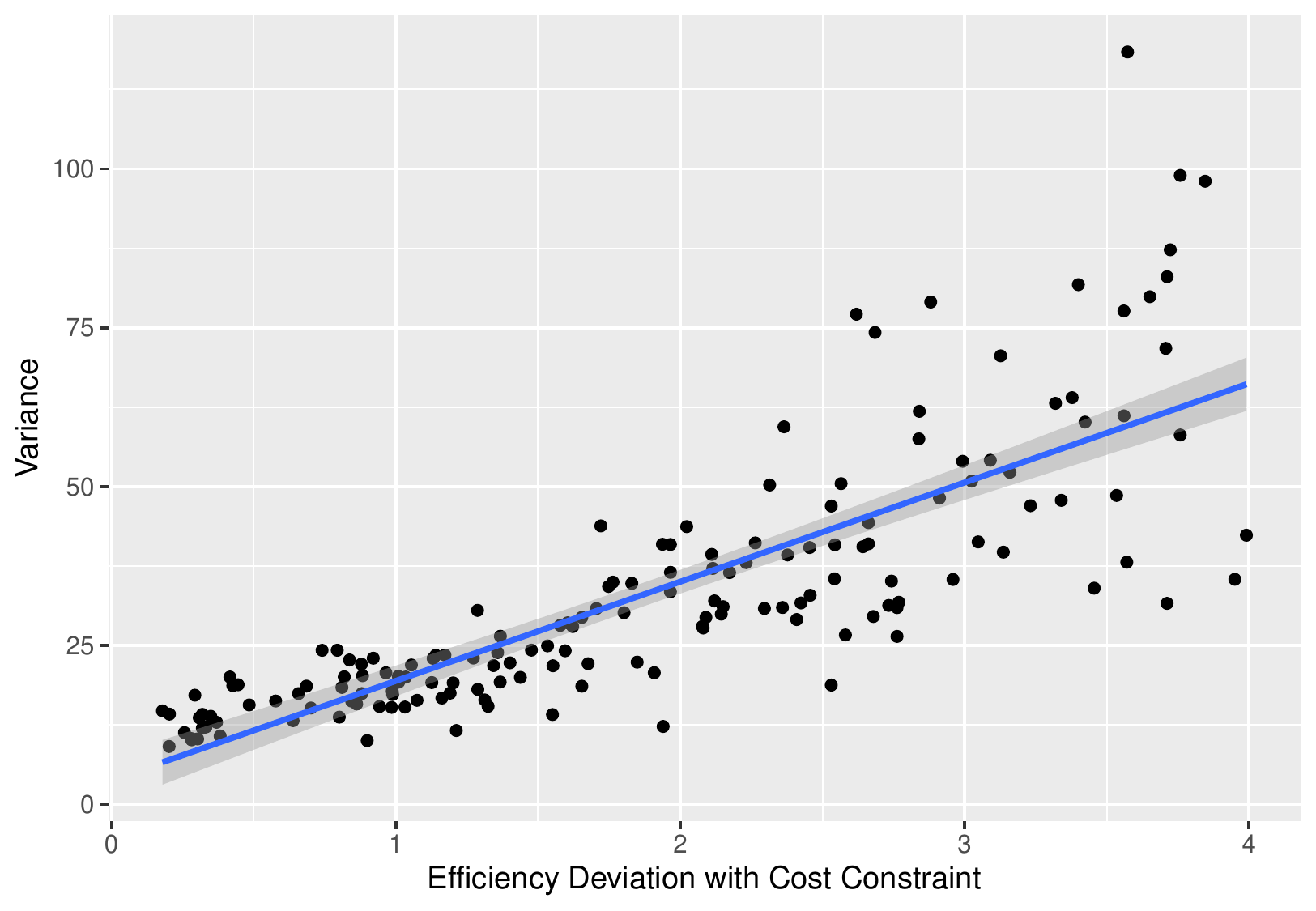}
    \caption{How the the deviation metric of experiment samples with cost constraint 
    influences the variance of $\hat{\tau}$.}
    \label{fig:simc}
\end{figure}

\subsection{Real Data Illustration}
We use the well-cited Tennessee Student/Teacher Achievement Ratio (STAR) experiment to assess how covariate distribution of experiment samples influences the estimation efficiency of average treatment effect for target cohort. STAR is a randomized experiment started in 1985 to measure the effect of class size on student outcomes, measured by standardized test scores. Similar to the exercise in \cite{kallus2018removing}, we focus a binary treatment: $T=1$ for small classes (13-17 pupils), and $T=0$ for regular classes (22-25 pupils). Since many students only started the study at first grade, we took as treatment their class-type at first grade. The outcome $Y$ is the sum of the listening, reading, and math standardized test at the end of first grade. We use the following covariates $X$ for each student: student race ($X_1 \in \{ 1, 2 \}$) and school urbanicity ($X_2 \in \{1,2,3,4\}$). We exclude units with missing covariates. The records of 4584 students remain, with 1733 assigned to treatment (small class, $T=1$), and 2413 to control (regular size class, $T=0$). Before analysis we fill the missing outcome by linear regression based on treatments and two covariates so that both two potential outcomes $Y_0$ and $Y_1$ for each student are known.

We simulate 500 candidate experiment sample allocations. For each allocations, we select $n_1 = 500$ experiment units from the dataset with probability
$$
    \Pr(S=1 \mid X=x) = \frac{e^{p_{x_1x_2}}}{e^{p_{11}} + e^{p_{12}} + e^{p_{13}}+ e^{p_{14}}+ e^{p_{21}}+ e^{p_{22}}+ e^{p_{23}}+ e^{p_{24}}}.
$$
where $p_{11}, p_{12}, p_{13}, p_{14}, p_{21}, p_{22}, p_{23}, p_{24}$ are i.i.d. samples drawn from standard normal distribution. We can then compute the real covariate distribution $f_1(x)$ and the repressiveness of the experiment samples $\mathcal{D}(f_1) = \var_1 (f_1^*(X) / f_1(X))$. To estimate the efficiency of average treatment effect estimator, we conduct 200 experiments for each design. In each experiment, the treatment follows a Bernoulli distribution with probability $1733/4584 = 0.378$. The simulation result is shown in Fig~\ref{fig:real}. The relationship between the variance and the deviation metric $\mathcal{D}(f)$ can be fit into a line, which is consistent with our result that $\var(\hat \tau) \propto \mathcal{D}(f_1)$.

\begin{figure}
    \centering
    \includegraphics[width = 0.5\textwidth]{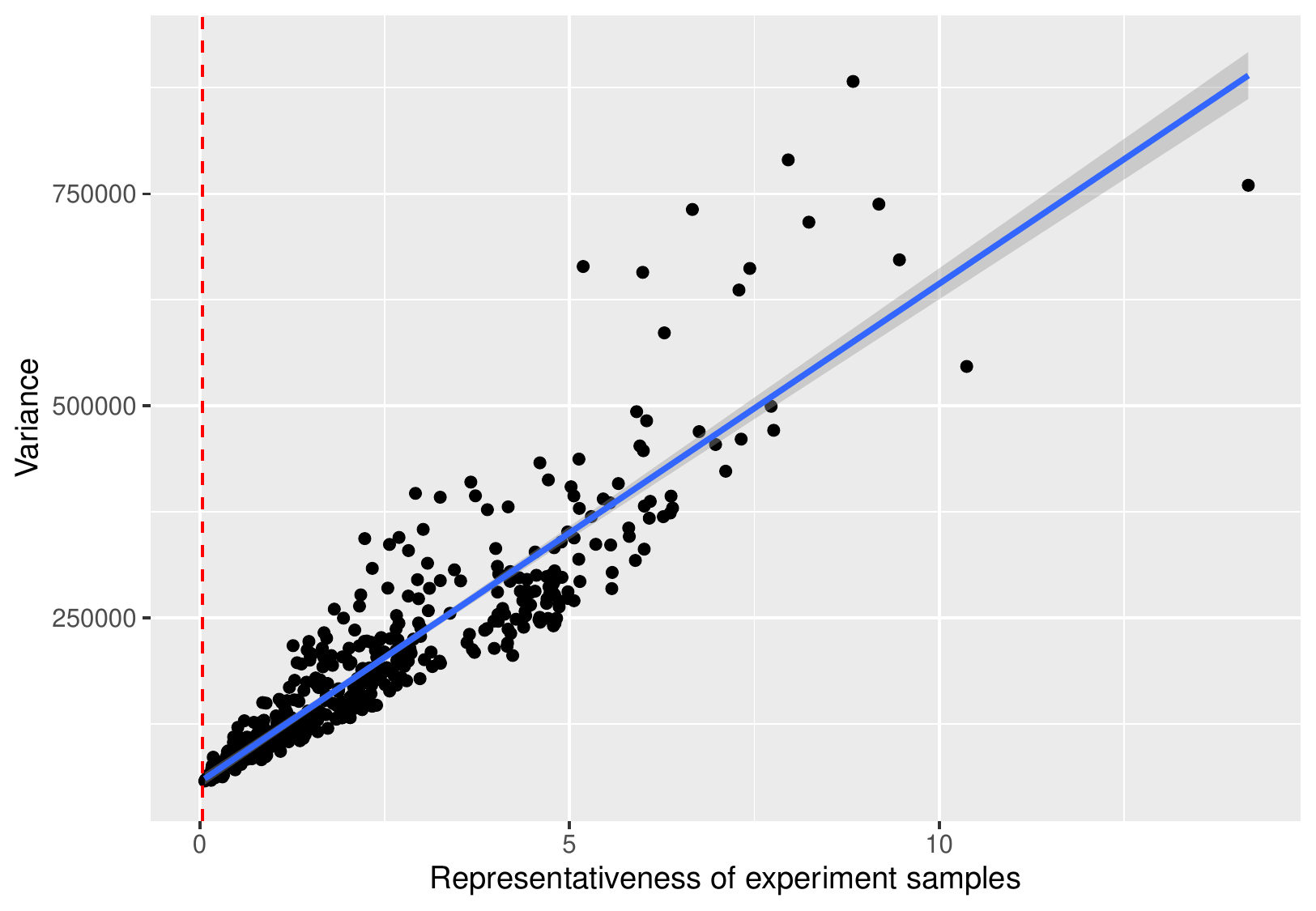}
    \caption{How the deviation metric of experiment samples $\mathcal{D}$ correlates with the estimated variance of the $\hat{\tau}$. The red line marks the deviation metric $\mathcal{D}$ of a trial selected following the na\"ive strategy.}
    \label{fig:real}
\end{figure}


\section{Conclusion}
In this paper, we examine the common procedure of generalizing causal inference from an RCT to a target cohort. We approach this as a problem where we can combine an RCT with an observational data. The observational data has two roles in the combination: one is to provide the exact covariate distribution of the target cohort, and the other role is to provide a means to estimate the conditional variability of the causal effect by covariate values.

We give the expression of the variance of Inverse Propensity Sampling Weights estimator as a function of covariate distribution in the RCT. We subsequently derive the variance-minimizing optimal covariate allocation in the RCT, under the constraint that the size of the trial population is fixed.  Our result indicates that the optimal covariate distribution of the RCT does not necessarily follow the exact distribution of the target cohort, but is instead adjusted by the conditional variability of potential outcomes. Practitioners who are at the design phase of a trial can use the optimal allocation result to plan the group of participants to recruit into the trial. \par
We also formulate a deviation metric quantifying how far a given RCT allocation is from optimal. The advantage of this metric is that it is proportional to the variance of the final ATE estimate so that when presented with several candidate RCT cohorts, practitioners can compare and choose the most efficient RCT according to this metric.\par
The above results depend on the estimation of conditional variability of the causal effect by covariate values, which remains unknown. We propose to estimate it using the observational data and outline mild assumptions that needs to be met. 
In reality, practitioners usually have complex considerations when designing a trial, for instance cost constraints and precision requirements. We develop variants of our main results to apply in such practical scenarios. Finally, we use two numerical studies to corroborate our theoretical results.  


\bibliographystyle{abbrvnat}
\bibliography{refs}

\newpage
\appendix

\section{Proof}
To be continued.

\begin{lemma}
\label{lem::minimum}
Let $f(x) = \sum_{i=1}^n a_i^2/x_i$, where $x = (x_1, \ldots, x_n)$ and $a = (a_1, \ldots, a_n)$ s.t. $ x, a \geq 0$ and $\sum_{i=1}^{n} x_i = 1$. Function $f(x)$ reaches its minimum when $x_i = a_i / (\sum_{i=1}^{n} a_i)$ for $i = 1, \ldots, n$.
\end{lemma}

\begin{proof}[Proof of Lemma~\ref{lem::minimum}]
Note that $x_n = 1 - \sum_{i=1}^{n-1} x_i$, for $i = 1, \ldots, n-1$, we have
\begin{eqnarray*}
    \frac{\partial f(x)}{\partial x_i} &=& -\frac{a_i^2}{x_i^2} + \frac{a_n^2}{(1 - \sum_{j =1}^{n-1} x_j)^2},\\
    \frac{\partial^2 f(x)}{\partial x_i^2} &=& \frac{2 a_i^2}{x_i^3} + \frac{2 a_n^2}{(1 - \sum_{j =1}^{n-1} x_j)^3} \geq 0.
\end{eqnarray*}
Given that $x$ is positive, by setting ${\partial f(x)}/{\partial x_i} = 0$, we can get $x_i = a_i / (\sum_{i=1}^{n} a_i)$ for $i = 1, \ldots, n$. This, together with ${\partial^2 f(x)}/{\partial x_i^2} \geq 0$, ensures that $f(x)$ reaches its minimum when $x_i = a_i / (\sum_{i=1}^{n} a_i)$ for $i = 1, \ldots, n$.
\end{proof}

\begin{proof}[Proof of Theorem~\ref{thm::covreal}]
Given Theorem~\ref{con::clt} and the definition of $\hat \tau$, $\hat\tau_\k$, a simple calculation then gives $\var(\hat \tau)$ and $\vara(\hat \tau_\k)$. Note that $\Vert K \Vert_2$ is constant given kernel $K$, according to Lemma~\ref{lem::minimum}, we have 
\begin{eqnarray*}
    f_1^*(x_m) = \frac{f_0(x_m) \sigma_\psi(x_m)}{\sum_{j=1}^M f_0(x_j) \sigma_\psi(x_j) }.
\end{eqnarray*}

We then have
\begin{eqnarray*}
    \sum_{m=1}^M f_0^2(x_m)\frac{\sigma_\psi^2(x_m)}{f_1(x_m)}
    &=& \sum_{m=1}^M f_1(x_m) f_0^2(x_m)\frac{\sigma_\psi^2(x_m)}{f_1^2(x_m)}\\
    &=& \E_{1} \left \{ \left( f_0(X)\frac{\sigma_\psi(X)}{f_1(X)} \right)^2 \right \}\\
    &=& \var_{1} \left( \frac{f_0(X) \sigma_\psi(X)}{f_1(X)}\right)
    + \E_{1}^2 \left \{  f_0(X)\frac{\sigma_\psi(X)}{f_1(X)} \right \}\\
    &=& \var_{1} \left( \frac{f_0(X) \sigma_\psi(X)}{f_1(X)}\right)
    + \left( \sum_{m=1}^M f_1(x_m) f_0(x_m) \frac{\sigma_\psi(x_m)}{f_1(x_m)}   \right)^2\\
    &=& \left( \sum_{m=1}^M f_0(x_m) \sigma_\psi(x_m)  \right)^2 \times \left\{\mathcal{D}(f_1) + 1\right\},
\end{eqnarray*}
where $\E_1(\cdot) = \E_{X \mid S=1}(\cdot)$.

\end{proof}

\begin{proof}[Proof of Theorem~\ref{thm::asyocd}]
We first show that in experimental data, we have 
\begin{eqnarray*}
\sigma_\psi^2(x) = \frac{1}{e(x)} \var(Y^{(1)} \mid X=x) + \frac{1}{1-e(x)} \var(Y^{(0)} \mid X=x).
\end{eqnarray*}

Note that $\tau(x) = m^{(1)}(x) - m^{(0)}(x)$ and $T(1-T)=0$, together with the definition of $\sigma_\psi^2(x)$, they ensure that
\begin{eqnarray*}
\sigma_\psi^2(x)
&=& \E \left[ \frac{T^2 (Y-m^{(1)}(x))^2 }{e(x)^2} +  \frac{(1-T)^2 (Y-m^{(0)}(x))^2 }{(1-e(x))^2} - 2\frac{T(1-T) (Y-m^{(1)}(x))(Y-m^{(0)}(x)) }{e(x)(1-e(x))} \mid X=x, S=1 \right] \\
&=& \E \left[ \frac{T^2 (Y-m^{(1)}(x))^2 }{e(x)^2} +  \frac{(1-T)^2 (Y-m^{(0)}(x))^2 }{(1-e(x))^2} \mid X=x, S=1 \right]\\
&=& \Pr(T=1 \mid X=x, S=1)  \E \left[ \frac{T^2 (Y-m^{(1)}(x))^2 }{e(x)^2}  \mid X=x, S=1 \right]\\
&& + \Pr(T=0 \mid X=x, S=1)  \E \left[ \frac{(1-T)^2 (Y-m^{(0)}(x))^2 }{(1-e(x))^2}  \mid X=x, S=1 \right]\\
&=& e(x) \frac{\E \left[ T^2 (Y-m^{(1)}(x))^2   \mid X=x, S=1 \right]}{e(x)^2} + (1 - e(x)) \frac{\E \left[ (1-T)^2 (Y-m^{(0)}(x))^2   \mid X=x, S=1 \right]}{(1-e(x))^2}\\
&=& \frac{1}{e(x)} \var(Y^{(1)} \mid X=x, S=1) + \frac{1}{1-e(x)} \var(Y^{(0)} \mid X=x, S=1)\\
&=& \frac{1}{e(x)} \var(Y^{(1)} \mid X=x) + \frac{1}{1-e(x)} \var(Y^{(0)} \mid X=x).
\end{eqnarray*}
The last equation holds due to Assumption~\ref{assump::identifiability1}.

Central limit theorem and Assumption~\ref{cond::prop} then give
\begin{eqnarray*}
\hat \sigma^2_\psi(x) &\rightarrow& \frac{1}{e(x)}\var{(Y^{(1)} \mid X=x, T=1, S=0)} + \frac{1}{1-e(x)}\var{(Y^{(0)} \mid X=x, T=0, S=0)},\\
\hat \sigma^2_\psi(x) &\rightarrow& c \sigma^2_\psi(x).
\end{eqnarray*}
This, together with Slutsky's theorem, ensures
\begin{eqnarray*}
\hat f_1^*(x) \rightarrow f_1^*(x).
\end{eqnarray*}
\end{proof}
\end{document}